\documentclass[twoside]{article}
\pdfoutput=1
\usepackage{aistats2018}
\usepackage{texshade}
\usepackage[utf8]{inputenc} 
\usepackage[T1]{fontenc}    
\usepackage{lmodern}
\usepackage{graphicx}
\usepackage{subfigure}
\usepackage{url}            
\usepackage{booktabs}       
\usepackage{amsfonts}       
\usepackage{nicefrac}       
\usepackage{bbm}
\usepackage{microtype}      
\usepackage{footnotebackref}[2012/07/01]
\usepackage{tablefootnote}[2014/01/26]
\usepackage{amsthm}
\newtheorem{theorem}{Theorem}[]

\usepackage{algorithm}  
\usepackage{algorithmicx}  
\usepackage{algpseudocode}  
\usepackage{amsmath}  
\floatname{algorithm}{Algorithm}

\usepackage{natbib}

\usepackage{graphicx}
\usepackage{tikz}
\usepackage{bm}
\usepackage{amsfonts}
\usepackage{amsmath}
\usepackage{amssymb}
\usepackage{graphicx}
\usepackage{mathrsfs}

\newcommand{\R}{\mathbb{R}}

\DeclareMathOperator*{\rank}{rank}
\renewcommand{\vec}[1]{\ensuremath{\mathbf{#1}}}
\newcommand{\vecs}[1]{\ensuremath{\mathbf{\boldsymbol{#1}}}}
\newcommand{\mat}[1]{\ensuremath{\mathbf{#1}}}

\newcommand{\A}{\mat{A}}
\renewcommand{\P}{\mat{P}}
\renewcommand{\S}{\mat{S}}
\renewcommand{\H}{\mat{H}}
\newcommand{\x}{\vec{x}}
\newcommand{\Scal}{{\mathcal{S}}}
\newcommand{\Pcal}{{\mathcal{P}}}

\begin{document}

%

%

\twocolumn[

\aistatstitle{Nonlinear Weighted Finite Automata}

\aistatsauthor{Tianyu Li \And Guillaume Rabusseau \And Doina Preup}

\aistatsaddress{ McGill University \And McGill University \And McGill University} ]

\begin{abstract}
Weighted finite automata~(WFA) can expressively model functions defined over strings but are inherently linear models.
Given the recent successes of nonlinear models in machine learning, it is natural to wonder whether extending WFA to the nonlinear
setting would be beneficial.
In this paper, we propose a novel model of
neural network based nonlinear WFA model (NL-WFA) along with a learning algorithm.
Our learning algorithm is inspired by the \emph{spectral learning} algorithm for WFA and relies
on a nonlinear decomposition of the so-called Hankel matrix, by means of an auto-encoder network. 
The expressive power of NL-WFA and the proposed learning algorithm are assessed on both synthetic and real world data, showing 
that NL-WFA can lead to smaller model sizes and infer complex grammatical structures from data.
\end{abstract}

\section{Introduction}
Many tasks in natural language processing, computational biology or reinforcement learning, rely on estimating  functions
mapping sequences of observations to real numbers.
\emph{Weighted finite automata (WFA)} are finite state machines that allow one to succinctly represent such functions.
WFA have been widely used in many fields such as grammatical parsing~\citep{mohri1998dynamic}, sequence modeling and prediction~\citep{cortes2004rational} and bioinfomatics~\citep{allauzen2008sequence}.
A \emph{probabilistic WFA (PFA)} is a WFA satisfying some constraints that computes a probability distribution over strings;
PFA are expressively equivalent to \emph{Hidden Markov Models (HMM)}~\citep{dupont2005links}, which have been successfully applied in many tasks such as speech 
recognition~\citep{gales2008application} and human activity recognition~\citep{nazabal2015discriminative}.
Recently, the so-called \emph{spectral method} has been proposed as an alternative to EM based algorithms to learn
HMM~\citep{hsuspectral}, WFA~\citep{bailly2009grammatical}, predictive state representations~\citep{boots2011closing}, and related models. Compared to EM based methods, the spectral method
has the benefits of providing consistent estimators and reducing computational complexity.

Although WFA have been successfully applied in various areas of machine learning, they are inherently linear models: their computation boils down
to the composition of linear maps. Recent positive results in machine learning have shown that models based on composing nonlinear functions are both very
expressive and able to capture complex structure in data.
 For example, by leveraging the expressive power of deep convolutional neural networks in the context of reinforcement learning, agents can be trained to 
outperform humans in Atari games~\citep{mnih2013playing} or to defeat world-class go players~\citep{silver2016mastering}. Deep convolutional networks have also
recently led to considerable breakthroughs in computer vision~\citep{krizhevsky2012imagenet}, where they showed their ability to disentangle the complex structure 
of the data by learning a representation which
unfold the original complex feature space~(where the data lies on a low-dimensional manifold) into a representation space where the structure has been linearized.
It is thus natural to wonder to which extent introducing non-linearity in WFA could be beneficial. We will show that both these advantages of nonlinear models, namely
their expressiveness and their ability to learn rich representations, can be brought to the classical WFA computational model.

In this paper, we propose a nonlinear WFA model~(NL-WFA) based on neural networks, along with a learning algorithm. In contrast with WFA, the computation of a NL-WFA relies on \emph{successive compositions of nonlinear mappings}.
This model can be seen as an extension of dynamical recognizers~\citep{moore1997dynamical} --- which are in some sense a nonlinear
extension of deterministic finite automata --- to the quantitative setting.
In contrast with the training of recurrent neural networks~(RNN), our learning algorithm does not rely on back-propagation through time. It is inspired
by the spectral learning algorithm for WFA, which can be seen as a two-step process: first find a low-rank factorization of the so called \emph{Hankel matrix}
leading to a natural embedding of the set of words into a low-dimensional vector space, and then perform regression in this representation space to recover the transition matrices. 
Similarly, our learning algorithm first finds a nonlinear factorization of the Hankel matrix using an auto-encoder network, thus learning a rich nonlinear representation
of the set of strings, and then performs  nonlinear regression using a feed-forward network to recover the transition operators in the representation space.

\textbf{Related works.}
NL-WFA and RNN are closely related: their computation relies on the composition of nonlinear mappings directed by a sequence of observations. In this paper, we explore a somehow orthogonal direction to the recent RNN literature by trying to connect such models back with classical computational models from formal language theory. Such connections have been explored in the past in the non-quantitative setting with dynamical recognizers~\citep{moore1997dynamical}, whose inference has been studied in e.g.~\citep{pollack1991induction}. The ability of RNN to learn classes of formal languages has also been investigated, see e.g.~\citep{avcu2017subregular} and references therein. 
It is well know that predictive state representations (PSR)~\citep{littman2002predictive} are strongly related with WFA~\citep{thon2015links}. 
A nonlinear extension of PSR has been proposed for deterministic controlled dynamical systems in~\citep{rudary2004nonlinear}. More recently, building upon reproducing kernel Hilbert space embedding of PSR~\citep{boots2013hilbert}, non-linearity is introduced  into PSR using recurrent neural networks~\citep{downey2017predictive,venkatraman2017predictive}. 
One of the main differences with these approaches is that our learning algorithm does not rely on back-propagation through time and we instead investigate how the spectral learning method for WFA  can be beneficially extended to the nonlinear setting. 

\section{Preliminaries} 
We first introduce notions on weighted automata and the spectral learning method.

\subsection{Weighted finite automaton} Let $\Sigma^*$ denote the set of strings over a finite alphabet $\Sigma$ and let $\lambda$ be the empty word. A \emph{weighted finite 
automaton}~(WFA) with $k$ states is a tuple $A=\langle \bm{\alpha}_0, \bm{\alpha}_\infty, \{\textbf{A}_\sigma\}\rangle$ where $ \bm{\alpha}_0,\bm{\alpha}_\infty\in\R^k$ are the initial and
final weight vector respectively, and $\mat{A}_\sigma\in\R^{k\times k}$ is the transition matrix for
each symbol $\sigma\in\Sigma$. A WFA computes a function $f_A:\Sigma^*\to\R$ defined for each
word $x=x_1x_2\cdots x_n\in\Sigma^*$ by 
$$f_A(x)= \bm{\alpha}_0^{\top} \mat{A}_{x_1}\mat{A}_{x_2}\cdots\mat{A}_{x_n}\bm{\alpha}_\infty.$$
By letting $\mat{A}_x= \mat{A}_{x_1}\mat{A}_{x_2}\cdots\mat{A}_{x_n}$ for any word $x=x_1x_2\cdots x_n\in\Sigma^*$ we will 
often use the shorter notation $f_A(x)=\vecs{\alpha}_0^{\top} \mat{A}_{x}\vecs{\alpha}_\infty.$
A WFA $A$ with $k$ states is \emph{minimal} if its number of states is minimal, i.e., any WFA $B$ such that
$f_A=f_B$ has at least $k$ states.
A function $f:\Sigma^*\to\R$ is \emph{recognizable} if it can be computed by a WFA. In this case  
the \emph{rank} of $f$ is the number of states of a minimal WFA computing $f$. If $f$ is not 
recognizable we let $\rank(f)=\infty$.

\subsection{Hankel matrix}
The \emph{Hankel matrix} $\mat{H}_f\in\R^{\Sigma^*\times\Sigma^*}$ associated with a function  $f:\Sigma^*\to\R$
is the bi-infinite matrix with entries $(\mat{H}_f)_{u,v}=f(uv)$ for all words $u,v\in\Sigma^*$.
The spectral learning algorithm for WFA relies on the following fundamental relation between the rank of $f$ and the rank of the Hankel matrix $\mat{H}_f$~\citep{carlyle1971realizations,fliess1974matrices}:
\begin{theorem}
For any $f:\Sigma^*\to\R$, $\rank(f)=\rank(\mat{H}_f)$.
\end{theorem}
In practice, one deals with finite sub-blocks of the Hankel matrix. 
Given a basis $\mathcal{B}=(\mathcal{P}, \mathcal{S})\subset \Sigma^*\times \Sigma^*$, where $\mathcal{P}$ is a set of \emph{prefixes} and $\mathcal{S}$ is a set of \emph{suffixes},
we denote the corresponding sub-block of the Hankel matrix by $\mat{H}_\mathcal{B}\in\R^{\mathcal{P}\times \mathcal{S}}$.
Among all possible basis, we are particularly interested in the ones with the same rank as $f$. We say that a basis is \emph{complete} if $rank(\textbf{H}_\mathcal{B})=rank(f)=rank(\textbf{H}_f)$.

For an arbitrary basis $\mathcal{B}=(\mathcal{P}, \mathcal{S})$, we define its \emph{p-closure} by $\mathcal{B}^\prime=(\mathcal{P^\prime}, \mathcal{S})$, 
where $\mathcal{P^\prime}=\Pcal\cup \mathcal{P}\Sigma$.
 It turns out that a Hankel matrix over a p-closed basis can be partitioned 
into $|\Sigma| + 1$ blocks of the same size~\citep{balle2014spectral}: 
\begin{displaymath}
\textbf{H}_{\mathcal{B}^\prime}^\top=[\textbf{H}_\lambda^\top| \textbf{H}_{\sigma_1}^\top|\cdots| \textbf{H}_{\sigma_{|\Sigma|}}^\top]
\end{displaymath}
where for each $\sigma\in\Sigma\cup\{\lambda\}$ the matrix $\mat{H}_{\sigma}\in\R^{\mathcal{P}\times \mathcal{S}}$ is defined by
$(\mat{H}_{\sigma})_{u, v}=f(u\sigma v)$.

\subsection{Spectral learning}
It is easy to see that the rank of the Hankel matrix $\mat{H}_f$ is upper bounded by the rank of $f$: if
$A= \langle \bm{\alpha}_0^{\top}, \bm{\alpha}_\infty, \{\textbf{A}_\sigma\}\rangle$ is a WFA with $k$ states computing $f$, then
$\mat{H}_f$ admits the rank $k$ factorization $\mat{H}_f=\mat{PS}$ where the  matrices 
$\mat{P}\in\R^{\Sigma^*\times k}$ and $\mat{S}\in\R^{k\times \Sigma^*}$ are defined by $\mat{P}_{u,:}=\bm{\alpha}_0^{\top}\mat{A}_u$
and $\mat{S}_{:,v}= \mat{A}_v\bm{\alpha}_\infty$ for all $u,v\in\Sigma^*$. Moreover, one can check that 
$\mat{H}_{\sigma}=\mat{P}\mat{A}_\sigma\mat{S}$ for each $\sigma\in\Sigma$.
The spectral learning algorithm relies on the non-trivial observation that this construction can be reversed:
given any rank $k$ factorization $\mat{H}_\lambda=\mat{PS}$, the WFA  
$A=\langle \bm{\alpha}_0^{\top}, \bm{\alpha}_\infty, \{\textbf{A}_\sigma\}\rangle$ defined by
$$\bm{\alpha}_0^{\top}=\mat{P}_{\lambda,:},\ \ \bm{\alpha}_\infty=\mat{S}_{:,\lambda},\ \text{  and
}\textbf{A}_\sigma=\textbf{P}^+ \textbf{H}_\sigma \textbf{S}^+,$$ 
is a minimal WFA computing $f$~\citep[Lemma 4.1]{balle2014spectral}, where
$\mat{H}_\sigma$ for $\sigma\in\Sigma\cup{\lambda}$ denote the finite matrices defined above for a prefix closed complete
basis $\mathcal{B}$.

\section{Nonlinear Weighted Finite Automata}\label{sec:NLWFA}
The WFA model assumes that the transition operators $\mat{A}_\sigma$ are linear. It is natural to wonder whether this linear assumption
sometimes induces a too strong model bias~(e.g. if one tries to learn a function that is not recognizable by a WFA).  
Moreover, even for recognizable functions, introducing non-linearity
could potentially reduce the number of states needed to represent the function. Consider the following example: 
given a WFA $A=\langle \bm{\alpha}_0, \bm{\alpha}_\infty, \{\textbf{A}_\sigma\}\rangle$, the function $(f_A)^2:u\mapsto f_A(u)^2$
is recognizable and can be computed by the WFA $A^\prime=\langle \bm{\alpha}_0^\prime, \bm{\alpha}_\infty^\prime, \{\textbf{A}_\sigma^\prime\}\rangle$ 
with $\bm{\alpha}_0^\prime=\bm{\alpha}_0 \otimes \bm{\alpha}_0$, $\bm{\alpha}_\infty^\prime=\bm{\alpha}_\infty \otimes \bm{\alpha}_\infty$ 
and $\textbf{A}_\sigma^\prime=\textbf{A}_\sigma \otimes \textbf{A}_\sigma$, where $\otimes$ denotes Kronecker  product. 
One can check that if $\rank(f_A)=k$, then  $\rank(f_{A'})$ can be as large as $k^2$, but intuitively   the \emph{true dimension} of
the model is $k$ using non-linearity%
\footnote{By applying the spectral method on the component-wise square root of the Hankel matrix of $A'$, one would
recover the WFA $A$ of rank $k$.}.
These two observations motivate us to introduce \emph{nonlinear WFA} (NL-WFA).

\subsection{Definition of NL-WFA}
We will use the notation $\tilde{g}$ to stress that a function $g$ may be nonlinear. 
We define a NL-WFA $\tilde{A}$ of with k states as a tuple $\langle\bm{\alpha}_0, \tilde{G}_\lambda, \{\tilde{{G}}_{\sigma}\}_{\sigma\in\Sigma}\rangle$,
where $\bm{\alpha}_0\in\R^k$ is a vector of initial weights, $\tilde{{G}}_{\sigma}: \mathbb{R}^k \to \mathbb{R}^k$ is a transition function
for each  $\sigma\in\Sigma$ and $\tilde{G}_\lambda: \mathbb{R}^k \to \mathbb{R}$ is a termination function. 
A NL-WFA $\tilde{A}$ computes a function $f_{\tilde{A}}:\Sigma^*\to\R$ defined by
$$f_{\tilde{A}}(x)=\tilde{G}_{\lambda}(\tilde{G}_{x_t}(\cdots\tilde{G}_{x_2}(\tilde{G}_{x_1}(\vecs{\alpha}_0))\cdots))$$
for any word $ x= x_1 x_2\cdots x_t\in\Sigma^*$. 
Similarly to the linear case, we will sometimes use the shorthand notation $\tilde{G}_x = \tilde{G}_{x_t} 
\circ  \tilde{G}_{x_{t-1}} \circ \cdots \circ \tilde{G}_{x_1}$.
This nonlinear model can be seen as a generalization of dynamical recognizers~\citep{moore1997dynamical} to the quantitative setting. It is easy
to see that one recovers the classical WFA model by restricting the functions $\tilde{{G}}_{\sigma}$ and $\tilde{G}_\lambda$ to be linear.
Of course some restrictions on these nonlinear functions  have to be imposed in order to control the expressiveness 
of the model. In this paper, we consider nonlinear functions computed by neural networks.

\subsection{A Representation learning perspective on the spectral algorithm}
Our learning algorithm is inspired by the spectral learning method for WFA. In order to give some insights and further motivate our
approach, we will first  show how the spectral method can be interpreted as a representation learning scheme.

The spectral method can be summarized as a two-stages process consisting of a \emph{factorization step} and
a \emph{regression step}: first find a low rank factorization of the Hankel matrix and then perform regression
to estimate the transition operators $\{\A_\sigma\}_{\sigma\in\Sigma}$.

First focusing on the factorization step, let us observe that one can naturally embed the set of prefixes into the vector space $\R^\Scal$ by mapping each prefix $u$ to the corresponding row of the Hankel matrix $\H_{u,:}$. However, it is
easy to check that this representation is highly redundant when the Hankel matrix is of low rank.
In the factorization step of the spectral learning algorithm, the rank $k$ factorization $\mat{H}=\mat{PS}$ can be seen as finding a low dimensional representation $\mat{P}_{u, :}\in\R^k$ for each prefix $u$, from 
which the original \emph{Hankel representation} $\H_{u,:}$ can be recovered using  the linear map $\S$~(indeed $\H_{u,:} = \P_{u,:}\S$).
We can formalize this encoder-decoder perspective by defining two maps $\Psi_p: \mathcal{P} \mapsto \mathbb{R}^{k}$ and $\Psi_s: \mathbb{R}^k \mapsto \mathbb{R}^{\mathcal{S}}$ by
$\Psi_p(u)^\top = \P_{u,:}$ and $\Psi_s(\x)^\top = \vec{x}^\top \S$. One
can easily check that $\Psi_s(\Psi_p(u))^\top=\mat{H}_{u, :}$, which implies that 
$\Psi_p(u)$ encodes all the information sufficient to predict the value $f(uv)$ for any suffix $v\in\mathcal{S}$~(indeed $f(uv) =\Psi_p(u)^\top\mat{S}_{:,v}$).

The regression step of the spectral algorithms consists in recovering the matrices $\A_\sigma$ satisfying
$\H_\sigma = \P\A_\sigma\S$. From our encoder-decoder perspective,
this can be seen as recovering the compositional mappings
$\A_\sigma$ satisfying $\Psi_p(u\sigma)^\top =  \Psi_p(u)^\top \A_\sigma$ for each $\sigma\in\Sigma$. 

It follows from the previous discussion that non-linearity could be beneficially brought to WFA and into the spectral learning algorithm
 in two ways: first
by using nonlinear methods to perform the factorization of the Hankel matrix, thus discovering a potentially nonlinear
embedding of the Hankel representation, and second by allowing the compositional feature maps associated to each symbol to
be nonlinear.

\section{Learning NL-WFA}
Introducing non-linearity can be achieved in several ways. In this paper, we will use neural networks due to their ability to discover relevant nonlinear low-dimensional representation spaces 
and their expressive power as function approximators.

\subsection{Nonlinear factorization}
Introducing non-linearity in the factorization step boils down to finding two mappings $\Psi_p$ and $\Psi_s$ 
such that $\Psi_s(\Psi_p(u))=\mat{H}_{u,:}$ for any prefix $u \in \mathcal{P}$. 
Briefly going back to the linear case, one can check that if $\mat{H} = \mat{PS}$, then we have $\mat{H}_{u,:} = \mat{H}_{u,:}\mat{S}^+\mat{S}$ for each prefix $u$, implying that the 
encoder-decoder maps satisfy $\Psi_p(u)^\top = \mat{H}_{u,:}\mat{S}^+$ and $\Psi_s(\vec{x})^\top = \vec{x}^\top \mat{S}$.
Thus the factorization step can essentially be interpreted as finding an auto-encoder able to project down the Hankel representation $\H_{u,:}$ 
to a  low dimensional space while preserving the relevant information captured by $\H_{u,:}$.

How to extend the factorization step to the nonlinear setting should now appear clearly: by training an auto-encoder to learn a low-dimensional representation
of the Hankel representations $\H_{u,:}$, one will potentially unravel a rich representation of the set of prefixes from which a NL-WFA can be 
recovered. 

 Let $\tilde{\phi}:\mathbb{R}^{\mathcal{S}} \mapsto \mathbb{R}^k$ and $\tilde{\phi}':\R^k\to\R^\Scal$ be the encoder and decoder maps respectively.
 We will train the auto-encoder shown in Figure~\ref{fig:NN-factorization} (left) to achieve
$$\tilde{\phi}^{\prime}(\tilde{\phi}(\mat{H}_{u, :})) \simeq \H_{u,:}.$$
More precisely, if $\textbf{H} \in \mathbb{R}^{m \times n}$, the model is trained to map the original Hankel representation $\mat{H}_{u, :} \in \mathbb{R}^{n}$ of
each  prefix $u$ to a latent representation vector in $\mathbb{R}^k$, where $k \ll n$, and then map this vector back to the original representation $\mat{H}_{u, :}$.
This is achieved by minimizing the reconstruction error~(i.e. the $\ell_2$ distance between the original representation and its reconstruction).
Instead of linearly factorizing the Hankel matrix, we use an auto-encoder framework consisting of two networks, whose hidden layer activation functions are 
nonlinear%
\footnote{We use the (component-wise) $\mathrm{tanh}$ function in our experiments.}. 

\begin{figure}[t]
\begin{center}
  \includegraphics[width=0.4\textwidth]{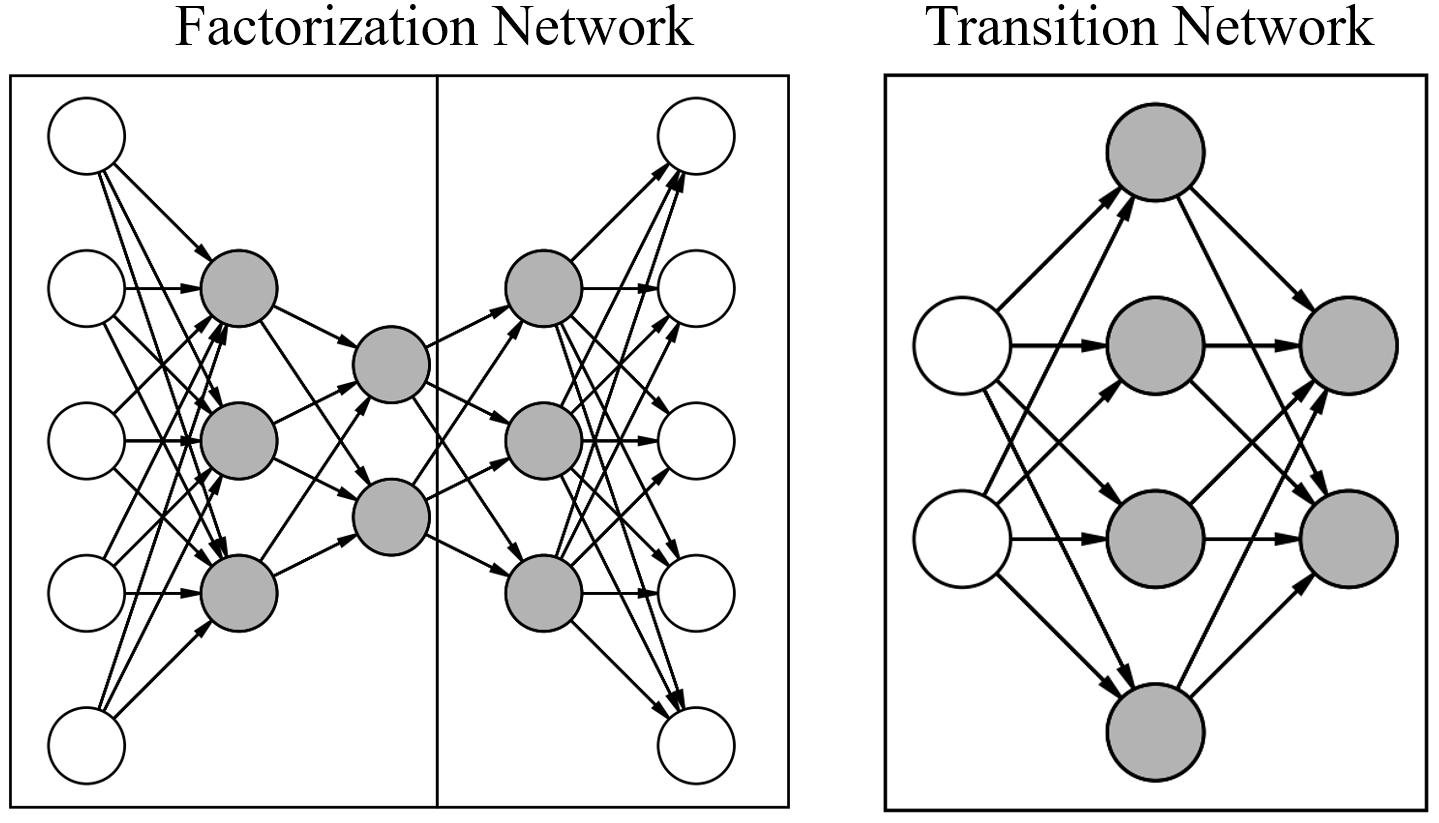}
  \caption{Factorization network and transition network:  grey units are nonlinear while white ones are linear.}
  \label{fig:NN-factorization}
\end{center}
\end{figure}

More precisely, if we denote the nonlinear activation function by $\theta$, and we let $\textbf{A}$, $\textbf{B}$, $\textbf{C}$, $\textbf{D}$ be the weights matrices from the left to the right of the neural net shown 
in Figure~\ref{fig:NN-factorization}~(left),
the function $\hat{f}:\mathbb{R}^n \to \mathbb{R}^n$ computed by the auto-encoder can be written as
$$\hat{f}=\tilde{\phi^{\prime}}\circ\tilde{\phi}: \mat(\mat{H})_{u, :}^\top \mapsto \theta(\theta(\theta(\textbf{H}_{u, :}^\top \mat{A})^\top\textbf{B})^\top\mat{C})^\top\textbf{D}$$
where the encoder-decoder functions $\tilde{\phi}: \mathbb{R}^n \to \mathbb{R}^k$ and $\tilde{\phi^{\prime}}: \mathbb{R}^k \to \mathbb{R}^n$ are
defined by $\tilde{\phi}(\vec{x})^\top=\theta(\theta(\vec{x}^\top \textbf{A})^\top\textbf{B})$ and $\tilde{\phi^{\prime}}(\vec{h})^\top=\theta({\vec{h}^\top}\textbf{C})^\top\textbf{D}$
for vectors  $\vec{x}\in \R^n, \vec{h}\in\R^k$.

It is  easy to check that if the activation function $\theta$ is the identity, one will exactly recover a rank $k$ factorization of the Hankel matrix, thus 
falling back onto the classical factorization step of the spectral learning algorithm. 

\subsection{Nonlinear regression}
Given the encoder-decoder maps  $\tilde{\phi}$ and $\tilde{\phi^{\prime}}$,
we then move on to recovering the transition functions. Recall that we wish to find the compositional feature maps $\tilde{G}_{\sigma}:\R^k\to\R^k$ for
each $\sigma$ satisfying $\Psi_p(u\sigma)=\tilde{G}_{\sigma}(\Psi_p(u))$ for all $u\in\Pcal$. 
Using the encoder map $\tilde{\phi}$ obtained in the factorization step, the mapping $\Psi_p$ can be written
as $\Psi_p(u) = \tilde{\phi}(\H_{u,:})$. 

In order to learn these transition maps, we will thus train one neural network for each symbol $\sigma$ to minimize the following squared error loss function
$$\sum_{u\in\Pcal}\lVert \tilde{G}_{\sigma}(\tilde{\phi}(\mat{H}_{u,:}))-\tilde{\phi}(\mat{H}_{u\sigma,:})\rVert^2.$$

The structure of the simple feed-forward network used to learn the transition maps is shown in Figure~\ref{fig:NN-factorization}~(right). 
Let $\mat{E}, \mat{F}$ be the two weights matrices, the function $\hat{g}: \mathbb{R}^k \to \mathbb{R}^k$ computed by this network can be written as
$$\hat{g}: \vec{h}^\top \mapsto \theta(\theta(\vec{h}^\top \mat{E})^\top \mat{F})$$
We want to point out that both  hidden units and  output units of this network are nonlinear. Since
this network will be trained to map between latent representations computed by the factorization network, the output
units of the transition network and the units corresponding to the latent representation in the factorization network should be
of the same nature to facilitate the optimization process.

\subsection{Overall learning algorithm}\label{sec:algo}
Let $(\Pcal,\Scal)\subset \Sigma^*\times \Sigma^*$ be a basis of suffixes and prefixes such that $\lambda\in\Pcal\cap\Scal$. Let $(\Pcal',\Scal)$ be its 
$p$-closure~(i.e. $\Pcal' = \Pcal\cup\Pcal\Sigma$) and let $m=|\Pcal'|,\ n=|\Scal|$. For reasons that will be clarified in the
next section, we assume that $\Pcal$ is prefix-closed~(i.e. for any $x\in\Pcal$, all prefixes of $x$ also belong to $\Pcal$).
The first step consists in building 
the estimate $\H\in\R^{m\times n}$ of the Hankel matrix  from the training data~(by using e.g. the empirical frequencies in the train set), where the rows of 
$\H$ are indexed by prefixes in $\Pcal'=\Pcal\cup \Pcal\Sigma$ and its columns by suffixes in $\Scal$. 
The learning algorithm for NL-WFA  then consists of two steps:
\begin{enumerate}
\item Train the factorization network to obtain a nonlinear decomposition of the Hankel matrix $\H$ through the
mappings $\tilde{\phi}:\R^n \to\R^k$ and $\tilde{\phi}':\R^k\to\R^n$ satisfying 
\begin{equation}\label{eq:factorization}
\tilde{\phi}'(\tilde{\phi}(\H_{u,:})) \simeq \H_{u,:}\ \ \text{for all }u\in\Pcal\cup\Pcal\Sigma.
\end{equation}

\item Train the transition networks for each symbol $\sigma\in\Sigma$ to learn the transition maps
$\tilde{G}_\sigma:\R^k\to\R^k$ satisfying
\begin{equation}\label{eq:transition}
\tilde{G}_{\sigma}(\tilde{\phi}(\mat{H}_{u,:})) \simeq \tilde{\phi}(\mat{H}_{u\sigma,:})\ \text{for all }u\in\Pcal.
\end{equation}
\end{enumerate}

The resulting NL-WFA is then given by $\tilde{A} = \langle \vecs{\alpha}_0, \tilde{G}_\lambda, \{\tilde{G}_\sigma\}_{\sigma\in\Sigma}\rangle$ where
$\bm{\alpha}_0=\tilde{\phi}(\mat{H}_{\lambda,:})$ and $\tilde{G}_\lambda$ is defined by
$$\tilde{G}_{\lambda}(\vec{x})=\bm{\lambda}^{\top} \tilde{\phi}^{\prime}(\vec{x}) \ \text{for all }\x\in\R^k$$
where $\bm{\lambda}$ is the one-hot encoding of the empty suffix $\lambda$.

%
%
%
%

\subsection{Theoretical analysis}

While the definitions of the initial vector $\vecs{\alpha}_0$ and termination function $G_\lambda$ given above may seem
\emph{ad-hoc}, we will now show that the learning algorithm we derived corresponds 
to minimizing an error loss function between $f_{\tilde{A}}(u)$ and the estimated value $\H_{u,\lambda}$ over
all prefixes in $\Pcal$. Intuitively, this means that our learning algorithm aims at minimizing the empirical squared error loss over the
training set $\Pcal\subset \Sigma^*$. More formally, we show in the following theorem that if both the factorization network and the transition
networks are trained to optimality~(i.e. they both achieve $0$ training error), then the resulting NL-WFA  exactly recovers the values given in
the first column of the estimate of the  Hankel matrix. 

\begin{theorem}\label{thm:zero.loss}
If the prefix set $\Pcal$ is prefix-closed and if equality holds in Eq.~\eqref{eq:factorization} and Eq.~\eqref{eq:transition}, then
the NL-WFA $\tilde{A} = \langle \vecs{\alpha}_0, \tilde{G}_\lambda, \{\tilde{G}_\sigma\}_{\sigma\in\Sigma}\rangle$, where
$\bm{\alpha}_0=\tilde{\phi}(\mat{H}_{\lambda,:})$ and $\tilde{G}_\lambda: \vec{x} \mapsto \bm{\lambda}^{\top} \tilde{\phi}^{\prime}(\vec{x})$,
is such that $f_{\tilde{A}}(u) = \H_{u,\lambda}$ for all $u\in\Pcal$.
\end{theorem}
\begin{proof}
We first show by induction on the length of a word $u=u_1u_2\cdots u_t \in\Pcal$ that 
$$\tilde{G}_u(\vecs{\alpha}_0) = \tilde{G}_{u_t}(\tilde{G}_{u_{t-1}}(\cdots \tilde{G}_1(\vecs{\alpha}_0)\cdots)) = \tilde{\phi}(\H_{u,:}).$$
If $u=\sigma\in \Sigma$, using the fact that $\lambda\in\Pcal$ we have $\tilde{G}_\sigma(\vecs{\alpha}_0) = \tilde{G}_\sigma(\tilde{\phi}(\H_{\lambda,:})) = 
\tilde{\phi}(\H_{\sigma,:})$ by Eq.~\eqref{eq:transition}. Now if  $u=u_1u_2\cdots u_t \in\Pcal$, we can apply the induction hypothesis on 
$u_1u_2\cdots u_{t-1}$ (since $\Pcal$ is prefix-closed) to obtain
$ \tilde{G}_u(\vecs{\alpha}_0) =  \tilde{G}_{u_t}(\tilde{G}_{u_1\cdots u_{t-1}}(\vecs{\alpha}_0)) = \tilde{G}_{u_t}(\tilde{\phi}(\H_{u_1\cdots u_{t-1},:}))=
\tilde{\phi}(\H_{u,:})$ by Eq.~\eqref{eq:transition}.

To conclude, for any $u\in\Pcal$ we have
$f_{\tilde{A}}(u) = \tilde{G}_\lambda(\tilde{G}_u(\vecs{\alpha}_0))
=  \tilde{G}_\lambda(\tilde{\phi}(\H_{u,:})) 
= \bm{\lambda}^{\top} \tilde{\phi}^{\prime}(\tilde{\phi}(\H_{u,:}))
= \H_{u,:} \bm{\lambda}  = \H_{u,\lambda}$
by Eq.~\eqref{eq:factorization}.

\end{proof}

Intuitively, it follows that the learning algorithm described in Section~\ref{sec:algo} aims at minimizing the following loss function
\begin{align*}
J(\tilde{\phi},\tilde{\phi'}, \{\tilde{G}_\sigma\}_{\sigma\in\Sigma}) 
&= 
\sum_{u\in\Pcal} ( \vecs{\lambda}^\top \tilde{\phi}^\prime (\tilde{G}_u(\tilde{\phi}(\H_{\lambda,:})) - \H_{u,\lambda})^2 \\
&=
\sum_{u\in\Pcal} (f_{\tilde{A}}(u) - \hat{f}(u))^2
\end{align*}

where $\hat{f}(u)$ is the estimated value of the target function on the word $u$, and where the NL-WFA $\tilde{A}$ is a function of
the encoder-decoder maps $\tilde{\phi},\tilde{\phi}^\prime$ and of the transition maps $\tilde{G}_\sigma$ 
as described in Section~\ref{sec:algo}.

Even though Theorem~\ref{thm:zero.loss} seems to suggest that our learning algorithm is prone to over-fitting, this is not the case. Indeed, akin to the linear
spectral learning algorithm, the restriction on the number of states of the NL-WFA~(which corresponds to the size of the latent representation layer
in the factorization network) induces regularization and enforces the learning process to discriminate between signal and noise~(i.e. in practice, the
networks will not achieve $0$ error due to the bottleneck structure of the factorization network). 

\subsection{Applying non-linearity independently in the factorization and transition networks}
We have shown that non-linearity can be introduced into the two steps of our learning algorithm. We can thus consider three variants of this
algorithm where we either apply non-linearity in the factorization step only, in the regression step only, or in both steps.
It is easy to check that these three different settings correspond to three different NL-WFA models depending on whether the termination
function only is nonlinear, the transition functions only are nonlinear, or both the termination and transition functions are nonlinear.  
Indeed, recall that that a NL-WFA $\tilde{A}$ is defined as a tuple $\tilde{A} = \langle\bm{\alpha}_0, \tilde{G}_\lambda, \{\tilde{{G}}_{\sigma}\}_{\sigma\in\Sigma}\rangle$.
If no non-linearity are introduced in the factorization network, the termination function will have the form 
$$\tilde{G}_\lambda:\x\mapsto \vecs{\lambda}^\top \tilde{\phi}'(\x) = \vecs{\lambda}^\top \mat{D}^\top\mat{C}^\top\x$$
(using the notations from the previous sections), which is linear. Similarly, if no non-linearity are used in the
transition networks, the resulting maps $\tilde{G}_\sigma$ will be linear. 

One may argue that only applying non-linearity in the termination function $\tilde{G}_{\lambda}$ would not lead to an expressive 
enough model. 
However, it is worth noting that in this case, after the nonlinear factorization step, even though the transition functions are linear
they are operating on a nonlinear feature space. This is similar in spirit to the kernel trick, where a linear model is learned in a feature
space resulting from a nonlinear transformation of the initial input space.
Moreover, if we go back to the example of the squared function $(f_A)^2$ for some WFA $A$ with $k$ states~(see beginning of Section~\ref{sec:NLWFA}),
even though $(f_A)^2$ may have rank up to $k^2$, one can easily build a NL-WFA with $k$ states computing $(f_A)^2$ where only the termination function 
is nonlinear.

\begin{figure}[t]
\begin{center}
  \includegraphics[width=0.5\textwidth]{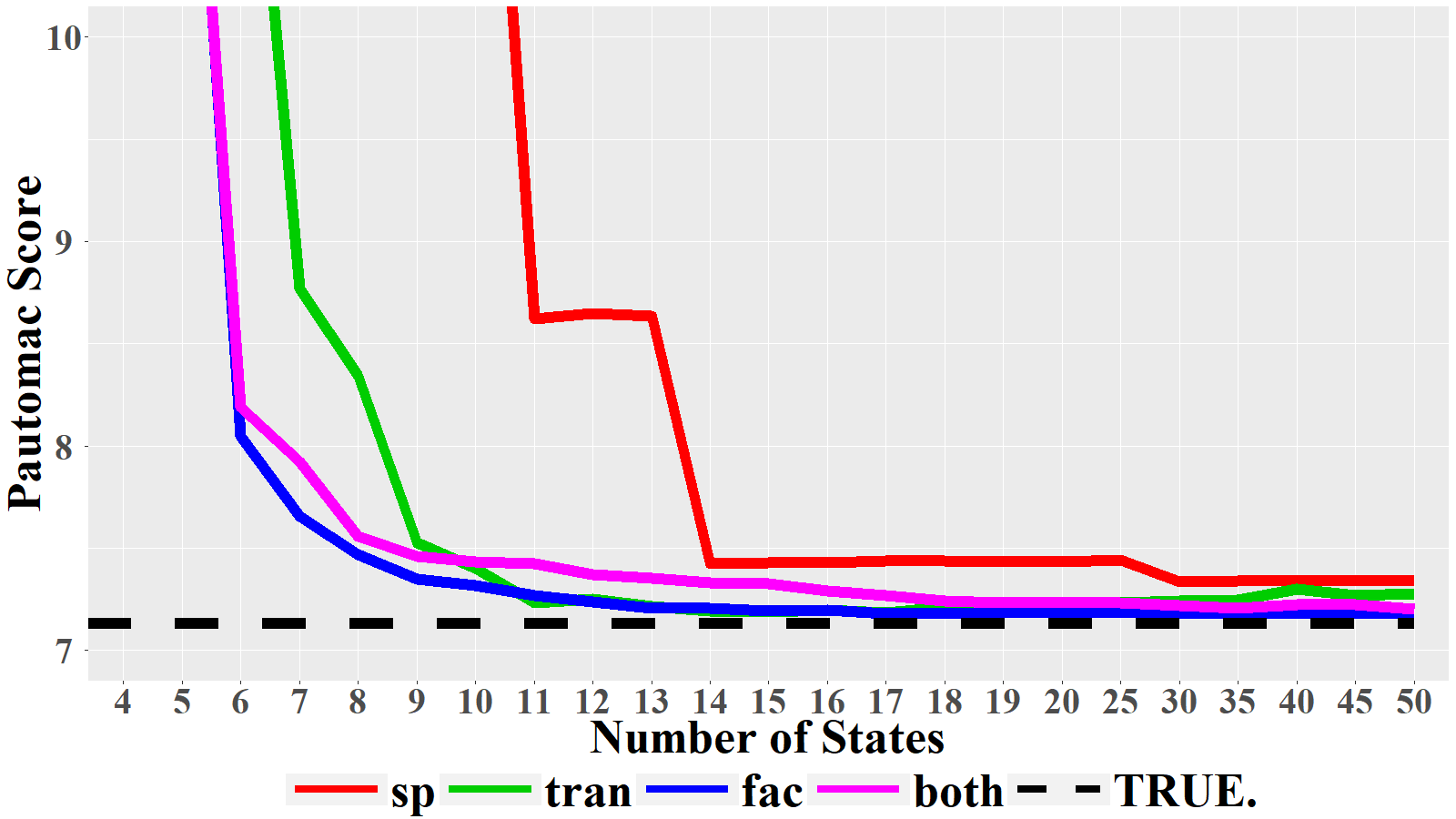}
  \caption{Pautomac score for the  Dyck language experiment for different  model sizes~(trained on a sample size of 20,000).}
  \label{fig:pauto1}
\end{center}
\end{figure}

\section{Experiments}
We compare the classical spectral learning algorithm with the three configurations of our neural-net based NL-WFA learning algorithms: 
applying non-linearity only in the factorization step~(denoted by \emph{fac.non}), only in the regression step~(denoted by \emph{tran.non}), and in both phases~(denoted by \emph{both.non}). 
We will perform experiments on a grammatical inference task~(i.e. learn a distribution over $\Sigma^*$ from samples drawn from this
distribution) with both synthetic and real data
\subsection{Metrics}
We use two metrics to evaluate the trained models on a test set: Pautomac score and word error rate.
\begin{itemize}
\item The Pautomac score was first proposed for the Pautomac challenge~\citep{verwer2014pautomac} and is defined by
$$\textit{Pauto}(M)=-2^{\sum_{x\in T}P_\ast(x)\log(P_M(x))}$$
where $P_M(x)$ is the normalized probability assigned to $x$ by the learned model and $P_\ast(x)$ is the normalized true 
probability~(both $P_M$ and $P_\ast$ are normalized to sum to $1$ over the test set $T$). Since the models returned by both our 
method and the spectral learning algorithm are not ensured to outputs positive values, while the logarithm of a negative value is not defined, we take the absolute values of all the negative outputs.

\item The word error rate~(WER) 
measures the percentage of incorrectly predicted symbols when, given each prefix of strings in the test set, the most likely next symbol is predicted.
\end{itemize}

\begin{figure}[t]
\begin{center}
  \includegraphics[width=0.5\textwidth]{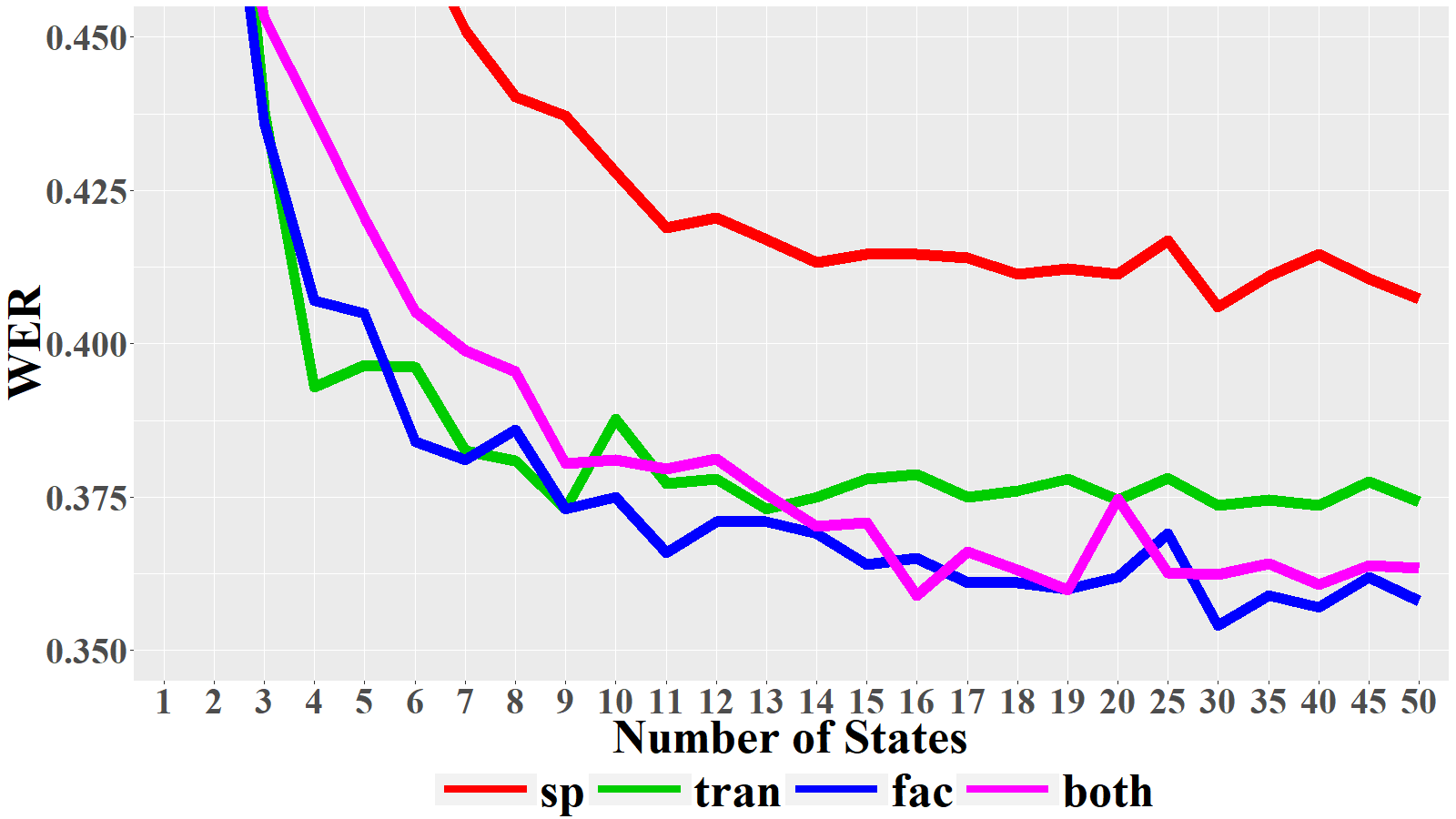}
  \caption{Word error rate for the  Dyck language experiment for different  model sizes~(trained on a sample size of 20,000).}
  \label{fig:wer1}
\end{center}
\end{figure}

\subsection{Synthetic data: probabilistic Dyck language}
For the synthetic data experiment, we generate data from a probabilistic Dyck language. Let $\Sigma=\{[, ]\}$, we consider the  language generated by the following probabilistic 
context free grammar
\begin{align*}
S &\to SS && \text{ with probability } 0.2\\
S &\to [S] && \text{ with probability } 0.4\\
S &\to [\ ] && \text{ with probability } 0.4\\
\end{align*}
i.e. starting from the symbol $S$, we draw one of the rules according to their probability and apply it to transform $S$ into the corresponding
right hand side; this process is repeated until no $S$ symbol are left.
One can check that this distribution will generate balanced strings of brackets. It is well known that this distribution cannot be computed by 
a WFA~(since its support is  a context free grammar). However, as a WFA can compute any distribution with finite support, it can model the restriction of this distribution to word of length less than
some threshold $N$. 
By using this distribution for our synthetic experiments, we want to showcase the fact that NL-WFA can lead to models with better predictive accuracy when the number of
states is limited and that they can better capture the complex structure of this distribution.

In our experiments, we use empirical frequencies in a training data set to estimate the 
 Hankel matrix $\mat{H}_\mathcal{B}\in\R^{1000\times 1000}$, where the p-closed basis $\mathcal{B}$ is obtained by selecting 
the $1,000$ most frequent prefixes and suffixes in the training data. 
We first assess the ability of NL-WFA to better capture the structure in the data when the number of states is limited.
We compared the models for different model sizes $k$ ranging from $1$ to $50$, where $k$ is the number of states of the learned
WFA and NL-WFA.
For the latter, we used a three hidden layers structure for the factorization network where the number of hidden units are set to $2k$, $k$ and $2k$. 
For the transition networks, we use a neural network with $2k$ hidden units\footnote{These hyper 
parameters are not finely tuned, thus some optimization might potentially improve the results.}. We used Adamax~\citep{kingma2014adam} with learning rate 0.015 and 0.001 respectively to train these two networks. 

\begin{figure}[t]
\begin{center}
  \includegraphics[width=0.5\textwidth]{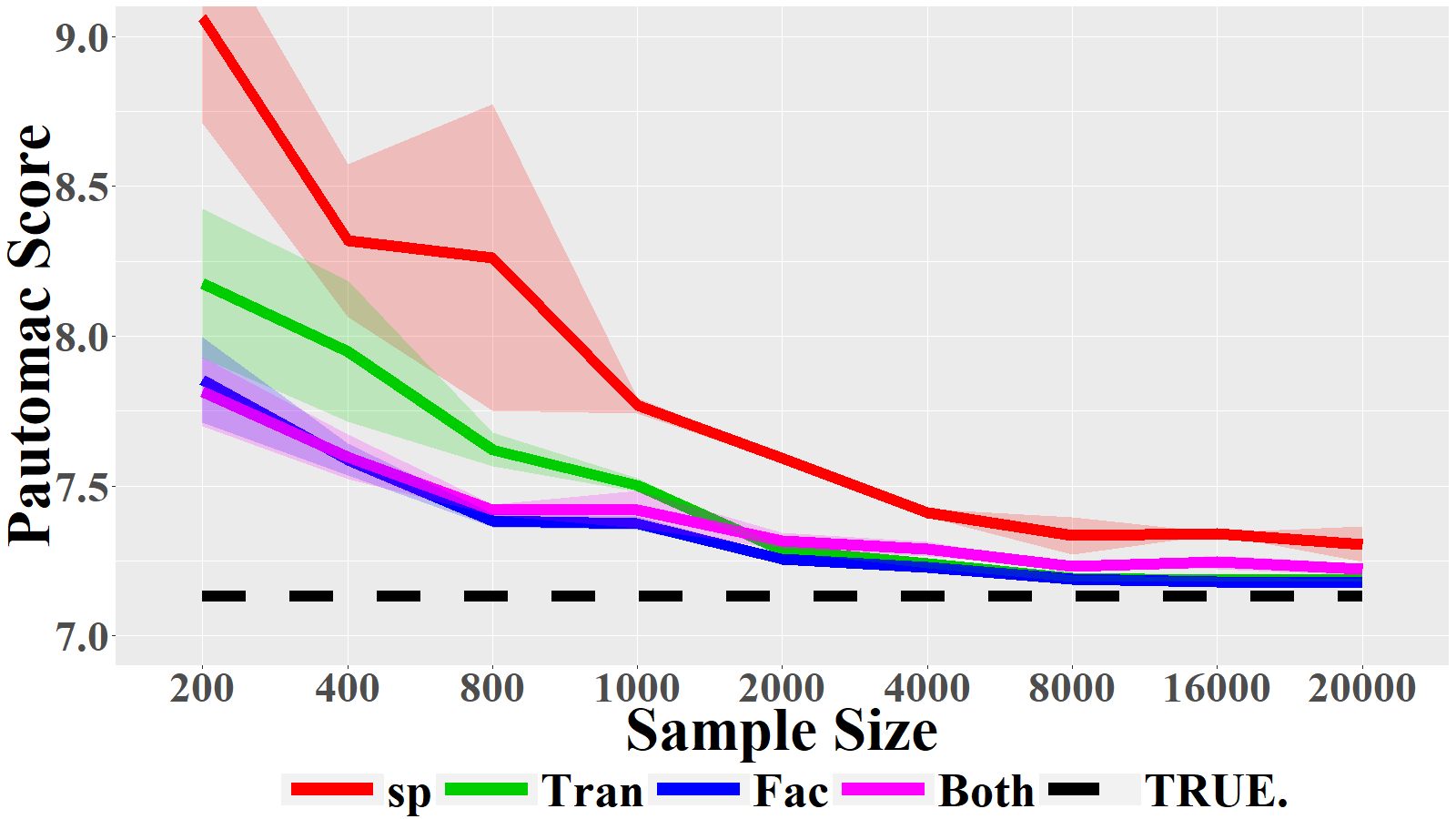}
  \caption{Average Pautomac score  for learning the Dyck language with different sample sizes.}
  \label{fig:pauto2}
\end{center}
\end{figure}

All models are trained on a training set of size $20,000$ and 
the Pautomac score and WER on a test set of size $250$ are reported  
in Figure~\ref{fig:pauto1} and~\ref{fig:wer1} respectively.
For both metrics, we see that NL-WFA gives better results for small model sizes. While NL-WFA and WFA tend to perform similarly for the Pautomac score
for larger model sizes, NL-WFA clearly outperforms WFA in terms of WER in this case.
This shows that including non-linearity can increase the prediction power of WFA by discovering the underlying nonlinear structure and can be beneficial when dealing with a small number of states.

We then compared the sample complexity of learning NL-WFA and WFA by training the different models on training set of sizes ranging from $200$ to $20,000$.
For all models the rank is chosen by cross-validation.
In Figure~\ref{fig:pauto2} and Figure~\ref{fig:wer2}, we show the performances for the four models on a test set of size $250$ by reporting the
average and standard deviation over $10$ runs of this experiment.
We can see that NL-WFA achieve better results on small sample sizes for the Pautomac score and consistently outperforms the linear model for all sample sizes for WER.
This shows that NL-WFA can use the training data more efficiently and again that the expressiveness of NL-WFA  is beneficial to this learning task.

\begin{figure}[t]
\begin{center}
  \includegraphics[width=0.5\textwidth]{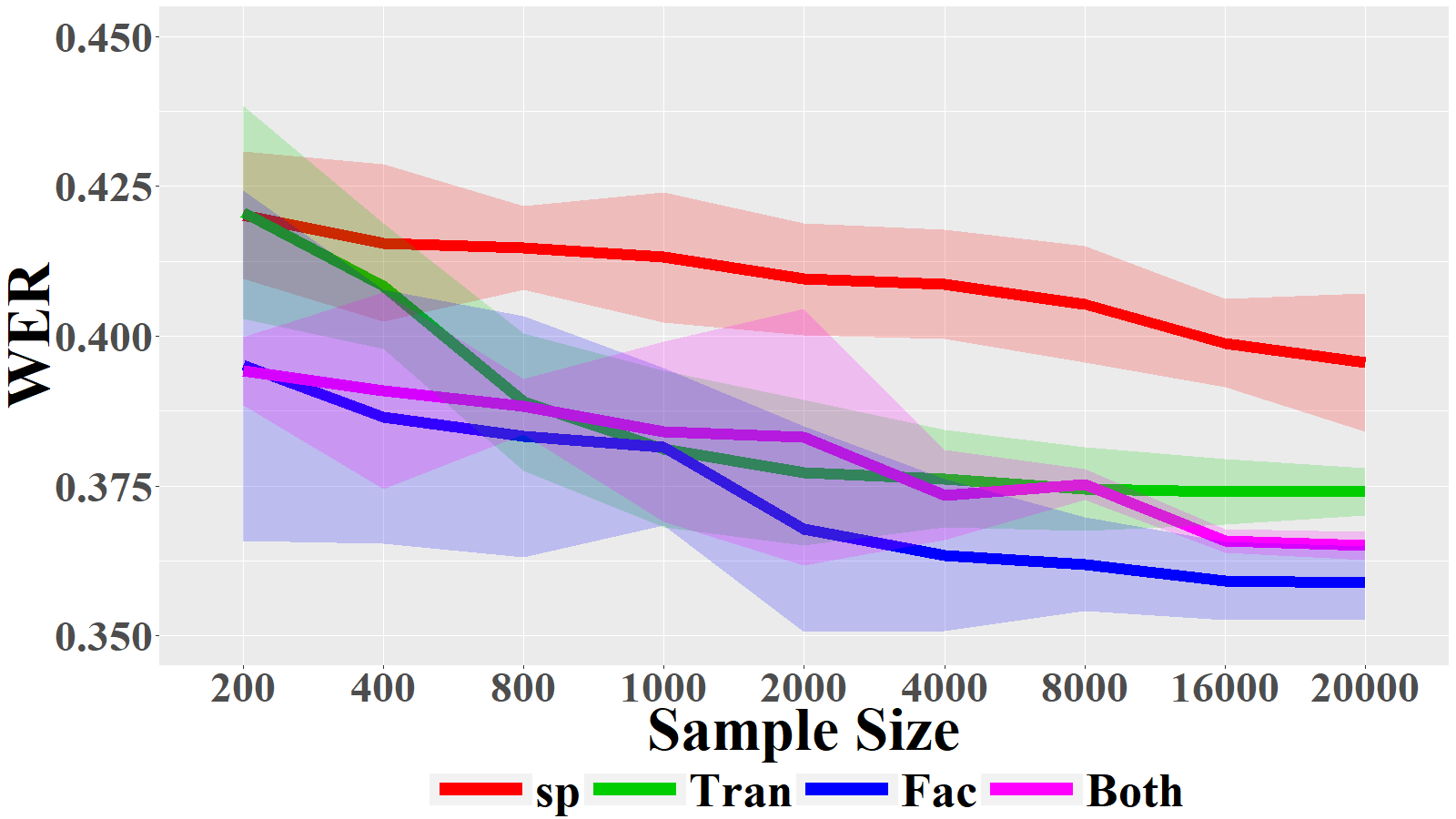}
  \caption{Average word error rate for learning the Dyck language with different sample sizes.}
  \label{fig:wer2}
\end{center}
\end{figure}
\subsection{Real data: Penn treebank}
The Penn Treebank~\citep{taylor2003penn} is a well known benchmark dataset for natural language processing. It consists of approximately
7 million words of part-of-speech tagged text, 3 million words
of skeletally parsed text, over 2 million words of text parsed for predicate argument
structure, and 1.6 million words of transcribed spoken text annotated
for speech disfluencies. In this experiment, we use a small portion of the Treebank dataset: the character level of English verbs
which was used in the SPICE challenge~\citep{balle2017results}. This dataset contains 5,987 sentences over an alphabet of 33 symbols as the training set. It also provides two test sets of size 750. We used one of the test sets as a validation set and then tested our models on the other.

For this experiment, the Hankel matrix $\H_{\cal B}$ is of size $3000\times 300$ where the prefixes and suffixes have been selected again by taking the most frequents in the training data.   We used a five layers factorization network where the layers are of size $4k$, $2k$, $k$, $2k$ and $4k$ respectively, where $k$ is the number of states of the NL-WFA. The structure of the transition networks is the same as in the previous experiment. For all models, the rank is selected using the validation set.

In Table~\ref{my-label}, we report the results for the two metrics on the test set. We can see that for both metrics, one of the NL-WFA models outperforms  linear spectral learning. Individually speaking, for modeling the distribution~(i.e. the perplexity metric) tran.non gives the best performances, while for the prediction task fac.non shows a significant advantage.

\begin{table}[H]
\centering
\resizebox{\columnwidth}{!}{
\begin{tabular}{lllll}
                & SP     & Tran.non & Fac.non & Both.non \\ \hline
log(Pauto)\tablefootnote{Since we do not have access to the true probabilities, $P_*$ is estimated using the  empirical frequencies in the test set.}       & 21.3807  & \textbf{12.2571}   & 13.8311  & 13.6604   \\ \hline
WER             & 0.8033 & 0.8841   & \textbf{0.7061}  & 0.8334  
\end{tabular}
}
\caption{Log perplexity and WER for real data}
\label{my-label}
\end{table}

\section{Discussion}
We believe that trying to combine models from formal languages theory~(such as weighted automata) and  models that have recently led to several successes in machine learning~(e.g. neural networks) is an exciting and promising line of research, both from the theoretical and practical sides. This work is a first step in this direction: 
we proposed a novel nonlinear weighted automata model along with a learning algorithm inspired by the spectral learning method for classical WFA. We showed that non-linearity can be introduced in two ways in WFA, in the termination function or in the transition maps, which directly translates into the two steps of our learning algorithm.

In our experiment, we showed on both synthetic and real data that (i) NL-WFA can lead to models with better predictive accuracy than WFA when the number of states is limited, (ii) NL-WFA are able to capture the complex underlying structure of challenging languages~(such as the Dyck language used in our experiments) and (iii) NL-WFA exhibit better sample complexity when learning on data with a complex grammatical structure.

In the future, we intend to investigate further the properties of NL-WFA from both the theoretical and experimental perspectives. For the former, one natural question is whether we could obtain learning guarantees for some specific classes of nonlinear functions. Indeed, one of the main advantages of the spectral learning algorithm is that it provides consistent estimators. While it may be  difficult  to obtain such guarantees when considering functions computed by neural networks, we believe that studying the case of more tractable nonlinear functions~(e.g. polynomials) could be very insightful. We also plan on thoroughly investigating connections between NL-WFA and RNN. From the practical perspective, we want to first tune the hyper-parameters for NL-WFA more extensively on the current datasets to potentially improve the results. In addition, we intend to run further experiments on real data and on different kinds of tasks beside language modeling~(e.g. classification, regression). Moreover, due to the strong connection between WFA and PSR, 
it will be very interesting to use NL-WFA in the context of reinforcement learning.

It is worth mentioning that the spectral learning algorithm cannot straightforwardly be  used to learn functions that are not probability distributions. Indeed, while it makes sense in the probabilistic setting to fill the entries corresponding to words that are not in the training data to $0$ in the Hankel matrix, it is not clear how to fill these entries when one wants to learn a function that is not a probability distribution, e.g. in a regression task. One way to circumvent this issue is to first use matrix completion techniques to fill these missing entries before performing the low rank decomposition of the Hankel matrix~\citep{balle2012spectral}. In contrast, our learning algorithm can directly be applied to this setting by simply adapting the loss function of the factorization network~(i.e. simply ignore the missing entries in the loss function).

\bibliographystyle{plainnat} 
\bibliography{biblio.bib}
\end{document}